\documentclass[12pt]{article}
 
\usepackage[margin=1in]{geometry} 
\usepackage{amsmath,amsthm,amssymb}

\usepackage{graphicx}

\usepackage{mathtools}
\DeclarePairedDelimiter{\ceil}{\lceil}{\rceil}

\newenvironment{theorem}[2][Theorem]{\begin{trivlist}
\item[\hskip \labelsep {\bfseries #1}\hskip \labelsep {\bfseries #2.}]}{\end{trivlist}}

\newenvironment{corollary}[2][Corollary]{\begin{trivlist}
\item[\hskip \labelsep {\bfseries #1}\hskip \labelsep {\bfseries #2.}]}{\end{trivlist}}
 
\begin{document}
 
 
\title{On the Computation of 2-Dimensional Recurrence Equations}
\author{Giuseppe Natale
\date{July 11, 2017}}
 
\maketitle

\section*{Introduction}
A generic 2-dimensional recurrence problem can be defined as:
\begin{align}
w_{i,j} &= a_{i,j} w_{i-1,j} + b_{i,j} w_{i,j-1} + c_{i,j}  & \forall (i,j) &\in \{1,..,n\}\times\{1,..,n\} \label{eq:2d} \\
w_{0,j} &= c_{0,j}                                              & \forall j &\in \{1,..,n\} \nonumber \\
w_{i,0} &= c_{i,0}                                              & \forall i &\in \{1,..,n\} \nonumber
\end{align}
It is quite easy to notice that the computational pattern induced by the dependence relations is a wavefront pattern. A visualization of such pattern is also shown in Figure \ref{fig:wf}. Indeed, wavefront parallelization is what is nowadays used to efficiently compute such recurrence problems, as it is able to reduce the computation time from being proportional to $n^2$ to being proportional to $2n-1$.

\begin{figure}
  \centering
  \includegraphics[width=0.3\linewidth]{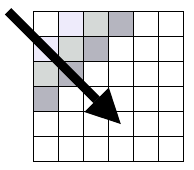}
  \caption{Wavefront computation pattern induced by the dependence relation in (\ref{eq:2d}).}
  \label{fig:wf}
\end{figure}

I here show how it is possible to systematically reduce a 2-dimensional recurrence problem of the above form into a mono-dimensional recurrence problem of the form:
\begin{align}
x_d &= h_d x_{d-1} + k_d  & \forall d &\in \{2,..,2n-1\} \label{eq:1d} \\
x_1 &= k_1 \nonumber
\end{align}
with $x_d$ and $k_d$ column vectors of length $n$, and $h_d$ a matrix of dimension $n\times n$.
This will prove that Kogge and Stone algorithm proposed in \cite{Kogge:1973:PAE:1638607.1639095} is applicable even to 2-dimensional recurrence relations (and probably also to 3+ dimensions following a similar approach). 

\section*{Discussion}
As previously said, the wavefront computational pattern enforced by the data dependencies can be exploited to perform parallelization and make the computation time proportional to $2n-1$ for an $n\times n$ grid. This is achieved by computing the grid points along the wavefront in parallel, in a way ``clustering'' the computation along the diagonal. I use the same ``clustering'' intuition to perform the proposed reduction, going from a 2-dimensional system $(i,j) \in \{1,..,n\}\times\{1,..,n\}$ to a mono-dimensional system $d \in \{1,..,2n-1\}$, where each $d$ corresponds to a diagonal -- i.e. a set of grid points of the same color as in Figure \ref{fig:wf} -- of the original 2-dimensional grid.
\begin{theorem}{1}
A 2-dimensional recurrence problem of the form of (\ref{eq:2d}) for an arbitrary $n$ can always be reduced to a mono-dimensional equivalent of the form of (\ref{eq:1d}), where $x_d$ and $k_d$ are column vectors of length $n$, and $h_d$ is a matrix of dimension $n\times n$. If the reduction provides less non-zero elements than $n$ for $x_d$ and $k_d$ (and less non-zero rows than  $n$ for $h_d$), then a number of 0s (or rows of 0s) is inserted after the last non-zero element until the vectors/matrices reach the required cardinality. This is just for consistency purposes and could be avoided in practical implementations.
\end{theorem}

\begin{proof}
The reduction is performed by relating the different $(i,j), \,\, \forall (i,j) \in \{1,..,n\}\times\{1,..,n\}$ of the 2-dimensional space of (\ref{eq:2d}) to the different $d, \,\, \forall d \in \{1,..,2n-1\}$ of the mono-dimensional space of (\ref{eq:1d}). In particular, each $(i,j)$ belonging to a diagonal -- i.e. a set of grid points of the same color as in Figure \ref{fig:wf} -- is grouped within the same $d$.

Indeed, we can define the different $x_d, \,\, \forall d \in \{1, .., 2n-1\}$ from the $w_{i,j} \,\, \forall (i,j) \in \{1,..,n\}\times\{1,..,n\}$:
\begin{align}
x_d &= 
\begin{cases}
    \begin{bmatrix}
        w_{1, d} \\
        w_{2, d-1} \\
        \vdots \\
        w_{1+i,d-i} \\
        \vdots \\
        w_{d-1, 2} \\
        w_{d, 1}
    \end{bmatrix} & \text{if } 1\leq d\leq n, \quad \forall i \in \{0,..,d-1\}\\ \label{eq:x_def} \\
    \begin{bmatrix}
        w_{d+1-n, n} \\
        w_{d+2-n, n-1} \\
        \vdots \\
        w_{(d+1-n) + i, n - i} \\
        \vdots \\
        w_{n-1, d+2-n} \\
        w_{n, d+1-n}
    \end{bmatrix} & \text{if } n+1\leq d\leq 2n-1, \quad \forall i \in \{0,..,2n-d-1\}\\  
\end{cases}
\end{align}

\newpage

In this way, each $x_d$ represent a diagonal (a collection of grid points) of the original 2-dimensional space. The $x_d$ column vector has a fixed cardinality of $n$, because of how I defined (\ref{eq:1d}) (a number of 0s will be inserted after the last non-zero element given by the relation defined in (\ref{eq:x_def}) until the vector $x_d$ doesn't reach a cardinality of $n$). That said, the number of non-zero values inside each $x_d$ column-vector is:
\begin{flalign*}
|x_d| &= 
\begin{cases}
    d & \text{if } 1\leq d\leq n\\
    2n-d & \text{if } n+1\leq d\leq 2n-1\\  
\end{cases}
\end{flalign*}

Now that the relation between the unknowns of (\ref{eq:2d}) and its mono-dimensional equivalent (\ref{eq:1d}) is defined, we can define the $h_d$ and $k_d, \,\, \forall d \in \{1,..,2n-1\}$ of (\ref{eq:1d}) from the $a_{i,j}, b_{i,j}, c_{i,j}$ and the $c_{0,j}, c_{i,0}, \,\, \forall (i,j) \in \{1,..,n\}\times\{1,..,n\}$ of (\ref{eq:2d}).
\begin{align}
k_d &= 
\begin{cases}
    \begin{bmatrix}
        c_{1, d} + a_{1, d}c_{0, d}\\
        c_{2, d-1} \\
        \vdots \\
        c_{1+i,d-i} \\
        \vdots \\
        c_{d-1, 2} \\
        c_{d, 1} + b_{d, 1}c_{d, 0}
    \end{bmatrix} & \text{if } 1\leq d\leq n, \quad \forall i \in \{0,..,d-1\}\\ \label{eq:k_def} \\
    \begin{bmatrix}
        c_{d+1-n, n} \\
        c_{d+2-n, n-1} \\
        \vdots \\
        c_{(d+1-n) + i, n - i} \\
        \vdots \\
        c_{n-1, d+2-n} \\
        c_{n, d+1-n}
    \end{bmatrix} & \text{if } n+1\leq d\leq 2n-1, \quad \forall i \in \{0,..,2n-d-1\}\\  
\end{cases}
\end{align}

From the definition of $k_d$ in (\ref{eq:k_def}), it is evident that $x_1 = [w_{1,1},0, .., 0]^T$ is implicitly defined as $x_1 = k_1$ since $w_{1,1} = a_{1, 1}c_{0, 1} + b_{1, 1}c_{1, 0} + c_{1,1}$. This is indeed expected, as also defined in (\ref{eq:1d}), since $x_1$ is the known initial value of the recurrence relation. We can therefore define $h_d$ only $\forall d \in \{2,..,2n-1\}$, and we can leave $h_1$ as undefined (or $= 0$ for simplicity), since it is not needed.

\begin{align}
& \begin{cases}
h_1 = 0
\end{cases} \label{eq:h_def}  \\
& \begin{cases} 
\text{if}\,\, 2\leq d\leq n, \quad \forall i \in \{0,..,d-1\}\\
h_d = 
\begin{bmatrix}
    b_{1, d} & 0 & 0 & \dots & 0 & 0 & 0\\
    a_{2, d-1} & b_{2, d-1} & 0 & \dots & 0 & 0 & 0\\
    0 & a_{3, d-2} & b_{3, d-2} & \dots & 0 & 0 & 0\\
    \vdots & & \ddots & & & & \vdots\\
    0 & \dots & a_{1+i, d-i} & b_{1+i, d-i} & \dots & 0 & 0\\
    \vdots & & & & \ddots & & \vdots\\
    0 & 0 & 0 & \dots & a_{d-1, 2} & b_{d-1, 2} & 0 \\
    0 & 0 & 0 & \dots & 0 & a_{d, 1} & 0 \\
\end{bmatrix}
\end{cases} \nonumber \\
& \begin{cases}
\text{if}\,\, n+1\leq d\leq 2n-1, \quad \forall i \in \{0,..,2n-d-1\}\\
h_d = 
\begin{bmatrix}
    a_{d+1-n, n} & b_{d+1-n, n} & 0 & \dots & 0 & 0 & 0\\
    0 & a_{d+2-n, n-1} & b_{d+2-n, n-1} & \dots & 0 & 0 & 0\\
    \vdots & & \ddots & & & & \vdots\\
    0 & \dots & a_{(d+1-n)+i, n-i} & b_{(d+1-n)+i, n-i} & \dots & 0 & 0\\
    \vdots & & & & \ddots & & \vdots\\
    0 & 0 & 0 & \dots & a_{n-1, d+2-n} & b_{n-1, d+2-n} & 0 \\
    0 & 0 & 0 & \dots & 0 & a_{n, d+1-n} & b_{n, d+1-n} \\
\end{bmatrix} \end{cases} \nonumber
\end{align}
\end{proof}

\begin{corollary}{1}
The ability to always reduce a 2-dimensional recurrence problem of the form of (\ref{eq:2d}) to a mono-dimensional equivalent of the form of (\ref{eq:1d}), as per Theorem 1, imply that Kogge and Stone algorithm presented in \cite{Kogge:1973:PAE:1638607.1639095} is indeed applicable. Moreover, each row of a generic $x_d = h_d x_{d-1} + k_d \,\, \forall d \in \{2,..,2n-1\}$ is independent from the others and can indeed be computed in parallel. This proves that -- at least theoretically -- by applying Kogge and Stone algorithm and parallelizing the computation of each row of the different $x_d \,\, \forall d \in \{1,..,2n-1\}$, the recurrence relations can be computed in time proportional to $\ceil{log_2(2n-1)}$.
\end{corollary}

\bibliographystyle{unsrt}
\bibliography{bib}

\section*{Appendix: Example}
I here show the reduction from (\ref{eq:2d}) to (\ref{eq:1d}) for $n=4$ as an example.
\begin{align*}
x_1 &= \begin{bmatrix} w_{1,1} \\ 0 \\ 0 \\ 0 \end{bmatrix}
& x_2 &= \begin{bmatrix} w_{1,2} \\ w_{2,1} \\ 0 \\ 0 \end{bmatrix}
& x_3 &= \begin{bmatrix} w_{1,3} \\ w_{2,2} \\ w_{1,3} \\ 0 \end{bmatrix}
& x_4 &= \begin{bmatrix} w_{1,4} \\ w_{2,3} \\ w_{3,2} \\ w_{4,1} \end{bmatrix}\\
x_5 &= \begin{bmatrix} w_{2,4} \\ w_{3,3} \\ w_{4,2} \\ 0 \end{bmatrix}
& x_6 &= \begin{bmatrix} w_{3,4} \\ w_{4,3} \\ 0 \\ 0 \end{bmatrix}
& x_7 &= \begin{bmatrix} w_{4,4} \\ 0 \\ 0 \\ 0 \end{bmatrix}
\end{align*}

Therefore:
\begin{align*}
k_1 &= \begin{bmatrix} c_{1,1} + a_{1,1}c_{0,1} + b_{1,1}c_{1,0} \\ 0 \\ 0 \\ 0 \end{bmatrix}
& k_2 &= \begin{bmatrix} c_{1,2} + a_{1,2}c_{0,2} \\ c_{2,1} + b_{2,1}c_{2,0} \\ 0 \\ 0 \end{bmatrix}
& k_3 &= \begin{bmatrix} c_{1,3} + a_{1,3}c_{0,3} \\ c_{2,2} \\ c_{3,1} + b_{3,1}c_{3,0} \\ 0 \end{bmatrix}\\
k_4 &= \begin{bmatrix} c_{1,4} + a_{1,4}c_{0,4} \\ c_{2,3} \\ c_{3,2} \\ c_{4,1} + b_{4,1}c_{4,0} \end{bmatrix}
& k_5 &= \begin{bmatrix} c_{2,4} \\ c_{3,3} \\ c_{4,2} \\ 0 \end{bmatrix}
& k_6 &= \begin{bmatrix} c_{3,4} \\ c_{4,3} \\ 0 \\ 0 \end{bmatrix}\\
k_7 &= \begin{bmatrix} c_{4,4} \\ 0 \\ 0 \\ 0 \end{bmatrix}\\ \nonumber \\
h_2 &= \begin{bmatrix}
    b_{1,2} & 0 & 0 & 0 \\
    a_{2,1} & 0 & 0 & 0 \\
    0 & 0 & 0 & 0 \\
    0 & 0 & 0 & 0 \\
\end{bmatrix}
& h_3 &= \begin{bmatrix}
    b_{1,3} & 0 & 0 & 0 \\
    a_{2,2} & b_{2,2} & 0 & 0 \\ 
    0 & a_{3,1} & 0 & 0 \\ 
    0 & 0 & 0 & 0 \\
\end{bmatrix}\\
h_4 &= \begin{bmatrix}
    b_{1,4} & 0 & 0 & 0 \\
    a_{2,3} & b_{2,3} & 0 & 0 \\ 
    0 & a_{3,2} & b_{3,2} & 0 \\ 
    0 & 0 & a_{4,1} & 0 \\ 
\end{bmatrix}
& h_5 &= \begin{bmatrix}
    a_{2,4} & b_{2,4} & 0 & 0 \\
    0 & a_{3,3} & b_{3,3} & 0 \\ 
    0 & 0 & a_{4,2} & b_{4,2} \\ 
    0 & 0 & 0 & 0 \\
\end{bmatrix}\\
h_6 &= \begin{bmatrix}
    a_{3,4} & b_{3,4} & 0 & 0 \\
    0 & a_{4,3} & b_{4,3} & 0 \\ 
    0 & 0 & 0 & 0 \\
    0 & 0 & 0 & 0 \\
\end{bmatrix}
& h_7 &= \begin{bmatrix}
    a_{4,4} & b_{4,4} & 0 & 0 \\  
    0 & 0 & 0 & 0 \\
    0 & 0 & 0 & 0 \\
    0 & 0 & 0 & 0 \\
\end{bmatrix}
\end{align*}

 
\end{document}